\documentclass[journal,10pt,twocolumn,twoside]{IEEEtran}
\ifCLASSINFOpdf
   \usepackage[pdftex]{graphicx}
\else
\fi
\usepackage[cmex10]{amsmath}
\usepackage{amssymb}
\usepackage{color}
\usepackage{tikz}

\newtheorem{theorem}{Theorem}
\newtheorem{definition}{Definition}
\newtheorem{proposition}[theorem]{Proposition}
\newtheorem{lemma}[theorem]{Lemma}

\newtheorem{corollary}[theorem]{Corollary}

\newcommand{\PP}{\mathbb{P}}
\newcommand{\todo}[1]{\textcolor{red}{#1}}

\newcommand{\E}{\mathcal{E}}

\begin{document}
%
\title{Exact Computation of Kullback-Leibler Distance for Hidden Markov Trees and Models}
%
%
%


\author{Vittorio~Perduca
        and~Gregory~Nuel\\  \todo{\emph{The present work is currently undergoing a major revision; a new version will be soon updated. \\ Please do not use the present version.}}
\thanks{V. Perduca and G. Nuel are with the MAP5 - Laboratory of Applied Mathematics, Paris Descartes University and CNRS, Paris, France, email: \texttt{\{vittorio.perduca, gregory.nuel\}@parisdescartes.fr}.}}
\maketitle

\begin{abstract}
We suggest new recursive formulas to compute the \emph{exact} value of the Kullback-Leibler distance (KLD) between two general Hidden Markov Trees (HMTs). For homogeneous HMTs with regular topology, such as homogeneous Hidden Markov Models (HMMs), we obtain a closed-form expression for the KLD when no evidence is given. We generalize our recursive formulas to the case of HMMs conditioned on the observable variables. Our proposed formulas are validated through several numerical examples in which we compare the exact KLD value with Monte Carlo estimations.
\end{abstract}


%
\IEEEpeerreviewmaketitle

\section{Introduction}
%
%
%
%

\IEEEPARstart{H}{idden} Markov Models (HMMs) are a standard tool in many applications, including signal processing \cite{crouse1998wavelet,ephraim2002hidden}, speech recognition \cite{Rabiner89atutorial,silva2008upper} and biological sequence analysis \cite{Durbin1999}. Hidden Markov Trees (HMTs, also called ``dependence tree models''), generalize HMMs on tree topologies, and are used in different contexts. In texture retrieval applications, they model the key features of the joint probability density of the wavelet coefficients of real-world data \cite{DO02}.

In estimation and classification contexts it is often necessary to compare different HMMs (or HMTs) through suitable distance measures. A standard (asymmetric) dissimilarity measure between two probability density functions $p$ and $q$ is the \textit{Kullback-Leibler distance} defined as \cite{bishop2006}:
\begin{equation*}
\label{eq:KLD_def}
D(p||q) = \int p \log \frac{p}{q}.
\end{equation*}
An exact formula for the KLD between two Markov chains was introduced in \cite{rached2004kullback}. Unfortunately there is no such a closed-form expression for HMTs and HMMs, as pointed out by   
several authors \cite{do2003fast, silva2008upper, Mohammad_Sahraeian_2011}.



To overcome this issue, several alternative similarity measures were introduced for comparing HMMs. Recent examples of such measures are based on a probabilistic evaluation of the match between every pair of states \cite{Mohammad_Sahraeian_2011}, HMMs' stationary cumulative distribution \cite{Zeng20101550} and transient behavior \cite{Silva_Narayanan_2006}. Other approaches are discussed in \cite{xie2007posteriori,mohammad_2005_novel}.

When it is mandatory to work with the actual KLD there are only two possibilities: 1) Monte Carlo estimation; 2) various analytical approximations. The former approach is easy to implement but also slow and inefficient. With regards to the latter, Do \cite{do2003fast} provided an upper bound for the KLD between two general HMTs. Do's algorithm is fast because its computational complexity does not depend on the size of the data. Silva and Narayan extended these results in the case of left-to-right transient continuous density HMMs \cite{silva2005upper,silva2008upper}. Variants of Do's result were discussed to consider the emission distributions of asynchronous HMMs (in the context of speech recognition) \cite{liu2007divergence} and marginal distributions \cite{xie2005probabilistic}. 


In this paper, we provide recursive formulas to compute the \emph{exact} KLD between two general HMTs with no evidence. In the case of homogeneous HMTs with regular topology, we derive a closed-form expression for the KLD. In the particular case of homogeneous HMMs, this formula is a straightforward generalization of the expression given for Markov chains in \cite{rached2004kullback}. It turns out that the KLD expression we suggest is exactly the well known bound introduced in \cite{do2003fast}: as a consequence, the latter is not a bound but the actual value of the KLD. At last, we generalize our recursive formulas to compute the KLD between two HMMs conditioned on the observable variables. We validated our models by comparing the exact value of the KLD with Monte Carlo estimations in the following cases: 1) HMTs with no evidence;  2) HMMs with arbitrarily given evidence; 3) HMMs with no evidence. For comparison purposes, we experimented with the same sets of parameters as in the examples of  \cite{do2003fast}. 

\section{Hidden Markov Trees}

\subsection{The model}

In a HMT, each node is either a hidden variable $S_u$ or an observable variable $X_u$. Only hidden variables have children. We denote as $S_{\emptyset}$ the root of the tree and as $S_u$ the parent  of $X_u$,  see Figure~\ref{fig:tree}. The joint probability distribution over all the variables of the model factorizes as 
$$
\PP(X,S) = \PP(S_{\emptyset})\PP(X_{\emptyset}|S_{\emptyset})\prod_u\PP(S_u|S_{\text{parent}(u)})\PP(X_u|S_u).
$$

We denote each index $u$ by a (finite) concatenation of characters belonging to a given finite and nonempty ordered set $V=\{0,1,2,3,\ldots\}$. In particular, $u$ is a regular expression belonging to $\{\emptyset\}\cup V_1\cup V_2 \cup \ldots \cup V_{N-1}$, where $\emptyset$ is the null string, $V_{i+1} = \{ua|u\in V_i,\,a\in V\}$, and where $N$ is the \emph{tree depth}. Using this notation, $S_v$ is a children of $S_u$ if and only if there exists $a\in V$ such that $v=ua$. In the binary example, shown in Figure \ref{fig:tree}, $V=\{0,1\}$ and $N-1=2$. 

\begin{figure}[!t]
\centering
\begin{tikzpicture}[scale=0.61, ->, thick]
    \tikzstyle{every node}=[]
    \scriptsize
    \node {$S_\emptyset$}
    [sibling distance = 3.7 cm]
        child { 
            node {$S_0$} 
            [sibling distance = 1.8 cm]
            child{ node {$S_{00}$}
                        [sibling distance = 1.24 cm]
                         child{node{$X_{00}$}}
             }
            child{node {$X_0$}}
            child{node{$S_{01}$}
            	   [sibling distance = 1.24 cm]
                      child{node{$X_{01}$}}
            }
         }
        child {
            node {$X_{\emptyset}$}}
            child { 
                node {$S_1$}
                [sibling distance = 1.8 cm]
                child{ node {$S_{10}$} 
                           [sibling distance = 1.24 cm]
                           child{node{$X_{10}$}}
                }
                child{node{$X_1$}}
                child{ node {$S_{11}$}
                            [sibling distance = 1.24 cm]
                            child{node{$X_{11}$}}
                 }
           }
    ;
\end{tikzpicture}
\caption{Formalism for HMTs, binary tree with $N = 3$.}
\label{fig:tree}
\end{figure}
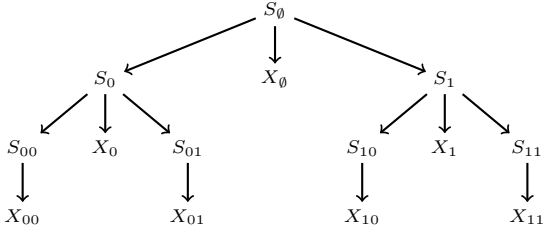

The parameters of the model are $\mathbb{P}(X_u=x | S_u=s)=e^u(s,x)$ (\textit{emissions}), $\mathbb{P}(S_{ua}=s | S_u=r)=\pi^{ua}(r,s)$ with $a\in V$ (\textit{transitions}), and $\mathbb{P}(S_{\emptyset}=s)=\mu(s)$ \footnote{For the sake of simplicity, we consider discrete variables, however it is straightforward to extend our results to the case of continuous variables, an example is in the Supplementary Material.}. We denote the set of parameters by $\theta$; $\PP_{\theta}$ denotes a probability distribution under $\theta$.

In the applications, we are often interested in $\PP(S|X=x)=\PP(X,S|\mathcal{E})$ where 
$\mathcal{E}=\{X=x \}$ is the \emph{evidence}.
Note that the notion of evidence can be generalized so to consider any subsets $\mathcal{X},\mathcal{S}$ of the sets of all possible outcomes of $X$ and $S$: $\mathcal{E}=\{X\in \mathcal{X},S\in \mathcal{S}\}$. For ease of notation, we consider only the cases when either no evidence is given or the evidence is  $\mathcal{E}=\{X=x\}$. We will explicitly develop the latter case only for HMMs, see Section \ref{subsec:hmm_evidence}, however it is easy to extend our results to the more general case of HMTs.


\subsection{Recursive formulas for exact KLD computation}

We derive recursive formulas for computing the exact Kullback-Leibler distance $D(\theta_1||\theta_0)=D(\PP_{\theta_1}(X,S)||\PP_{\theta_0}(X,S))$ between two HMTs having the same underlying topology $T$ and two distinct sets of parameters $\theta_1,\theta_0$. 

\begin{definition}
Given an index $u$ and $a\in V$, consider the variables $\{X_{ua-},S_{ua-}\}$ in the subtree $T_{ua}$ of $T$ rooted at $S_{ua}$ (e.g. $X_{01101}$ is in the subtree $T_{011}$ rooted at $S_{011}$, here $u=01$, $a=1$ and $-=01$). We define the \textit{inward quantity} $K_{ua\rightarrow u}(S_u)$ as the KLD between the conditional probability distributions of $\{X_{ua-},S_{ua-}\}$ given $S_u$, with parameters $\theta_1$ and $\theta_0$ respectively:
\begin{multline}
K_{ua\rightarrow u}(S_{u})=\\
D\left[\PP_{\theta_1}(X_{ua-},S_{ua-}|S_u)||\PP_{\theta_0}(X_{ua-},S_{ua-}|S_u)\right]
\end{multline}
where $ua-$ is reduced to $ua$ in the particular case when $X_{ua}$ is a leaf of the tree.
\end{definition}

Our first results are the following simple formulas that make it possible to compute the inward quantities and the (exact) KLD recursively (proofs in the Supplementary Material): 
\begin{theorem} 
\label{prop:tree_inward_initial}
\begin{multline}
K_{ua\rightarrow u}(S_u)=\\
\label{eq:tree_inward_initial}
\sum_{X_{ua},S_{ua}}\PP_{\theta_1}(X_{ua},S_{ua}|S_u)\left(\log\frac{\PP_{\theta_1}(X_{ua},S_{ua}|S_u)}{\PP_{\theta_0}(X_{ua},S_{ua}|S_u)} \right. + \\
 \left. \sum_{b\in V} K_{uab\rightarrow ua}(S_{ua})\right)
\end{multline}
with the convention that when $X_{ua}$ is a leaf, with $a\in V$, we have $K_{uab\rightarrow ua}(S_{ua})=0$ for each $b\in V$. Moreover,
\begin{multline}
D(\theta_1|| \theta_0)=\\
\label{eq:tree_kld_initial}
\sum_{X_{\emptyset},S_{\emptyset}}\PP_{\theta_1}(X_{\emptyset},S_{\emptyset})\left(\log\frac{\PP_{\theta_1}(X_{\emptyset},S_{\emptyset})}{\PP_{\theta_0}(X_{\emptyset},S_{\emptyset})} + \sum_{a\in V} K_{a\rightarrow\emptyset}(S_{\emptyset})\right).
\end{multline}
\end{theorem}

\subsection{Homogeneous trees with constant number of children}



When the tree is homogeneous and the nodes $S$ have the same number of children (e.g. when $T$ is binary as in Figure \ref{fig:tree}), Eqs. (\ref{eq:tree_inward_initial}) and (\ref{eq:tree_kld_initial}) can be further simplified:

\begin{corollary}
\label{prop:tree_hom_bin}
Suppose that the transition and emission probabilities are the same across the whole tree and each variable of type $S$ has exactly $C$ children of type $S$, then for each $a,a'\in V$: $K_{ua\rightarrow u}(S_u)=K_{ua'\rightarrow u}(S_u)$. In particular, if $X_{ua}$ is not a leaf, then for each $a\in V$:
\begin{multline}
K_{ua\rightarrow u}(S_u) = \\
\label{eq:in_hom_C}
k(S_u) + C\sum_{S_{u0}}\PP_{\theta_1}(S_{u0}|S_u)K_{u00\rightarrow u0}(S_{u0}), 
\end{multline}
where
$k(S_u=r)=D[\PP_{\theta_1}(X_0,S_0|S_{\emptyset}=r)||\PP_{\theta_0}(X_0,S_0|S_{\emptyset}=r)]=k(r)$.
Moreover
\begin{equation}
\label{eq:kld_hom_C}
D(\theta_1||\theta_0) = k_{\emptyset}+C\sum_{S_{\emptyset}}\PP_{\theta_1}(S_{\emptyset})K_{0\rightarrow \emptyset}(S_{\emptyset}),
\end{equation}
where $k_{\emptyset} = D[\PP_{\theta_1}(X_\emptyset,S_{\emptyset})||\PP_{\theta_0}(X_\emptyset,S_{\emptyset})]$.
\end{corollary}

By writing $\boldsymbol{\mu},\boldsymbol{k},\boldsymbol{\pi}$ as a row, a column and a square matrix respectively, we obtain the following closed formula: $D(\theta_1||\theta_0)=$
\begin{equation}
\label{eq:tree_kld_matrix}
k_{\emptyset} + \boldsymbol{\mu}_{\theta_1} (C \boldsymbol{I} + C^2\boldsymbol{\pi}_{\theta_1} + C^3\boldsymbol{\pi}^2_{\theta_1} + \ldots + C^{N-1}\boldsymbol{\pi}^{N-2}_{\theta_1}) \boldsymbol{k},
\end{equation} 
where $N$ is the depth of the tree, $\boldsymbol{I}$ the identity matrix and each node of type $S$ has exactly $C$ children of type $S$. 



\section{Hidden Markov Models}
\label{sec:hmm}

With reference to the notations used in the previous section, a HMM is a HMT in which each variable of type $S$ has only one child of type $S$ (i.e. $C=1$). In particular we can rename the variables so that $S=S_{1:N}$ is the hidden (Markov) sequence and $X=X_{1:N}$ is the sequence of observable variables.   
In the homogeneous case, the parameters of the model are $\mu(s) = \PP(S_1=s)$, $\pi(r,s)=\PP(S_i=s|S_{i-1}=r)$, $e(s,x)=\PP(X_i=x|S_i=s)$. 

\subsection{No Evidence}
When the variables $X_{1:N}$ are not actually observed (that is, there is no evidence in the model), all the formulas derived in the general case of trees continue to hold. Eq.~(\ref{eq:tree_kld_matrix}) gives the KLD between two homogeneous HMMs when no evidence is given:
\begin{equation}
\label{eq:kld_sum_hmm}
D(\theta_1||\theta_0) = k_\emptyset + \boldsymbol{\mu}_{\theta_1} ( \boldsymbol{I} + \boldsymbol{\pi}_{\theta_1} + \ldots +\boldsymbol{\pi}^{N-2}_{\theta_1}) \boldsymbol{k},
\end{equation}
where $k_\emptyset$ and $\boldsymbol{k}$ are defined exactly as in the previous section.
Note that this formula is a straightforward extension of the results for Markov chains proved in Theorem~1, \cite{rached2004kullback}. Moreover, it can be proved that the closed-form expression in Eq.~(\ref{eq:kld_sum_hmm}) is exactly the bound given in Eq.~(19)~\cite{do2003fast}, see the Supplementary Material for the details.

Let $\boldsymbol{\nu}$ be the stationary distribution of $\boldsymbol{\pi}_{\theta_1}$, then $\boldsymbol{\mu}_{\theta_1} \boldsymbol{\pi}_{\theta_1}^i \boldsymbol{k}$ converges towards $\boldsymbol{\nu} \boldsymbol{k}$ for large $i$. From Eq.~(\ref{eq:kld_sum_hmm}), by simply computing a Ces\`aro mean limit, we otbain the KLD \textit{rate}
\begin{equation}
\label{eq:kldr}
\bar{D}(\theta_1||\theta_0):=\lim_{N\rightarrow+\infty}\frac{D(\theta_1||\theta_0)}{N}  = \boldsymbol{\nu} \boldsymbol{k}.
\end{equation}
As observed in \cite{do2003fast}\footnote{$\boldsymbol{\nu} \boldsymbol{k}$ is exactly the bound for the KLD rate given in \cite{do2003fast}.}, $\boldsymbol{\nu} \boldsymbol{k}$ can be computed in constant time with $N$ whereas the exact closed formula of Eq.~(\ref{eq:kld_sum_hmm}) is computable in $O(N)$ with a direct implementation, or in $O(\log_2(N))$ with a more sophisticated approach (see the Supplementary Material for the details).

\subsection{$X$s observed}
\label{subsec:hmm_evidence}

Now we assume that the variables of type $X$ are actually observed, as it is often the case in practice. In particular, we consider the evidence $\mathcal{E} = \{X_{1:N} = x_{1:N}\}$ and we want to compute $D(\PP_{\theta_1}(X,S|\mathcal{E})\,||\, \PP_{\theta_0}(X,S|\mathcal{E}))=D(\PP_{\theta_1}(S|\mathcal{E})\,||\, \PP_{\theta_0}(S|\mathcal{E}))$.

For the sake of simplicity, we can denote the inward quantity indexed by $i+1\rightarrow i$ simply as $K^{\mathcal{E}}_i(S_i)$. Eqs.~(\ref{eq:tree_inward_initial}) and~(\ref{eq:tree_kld_initial}) become: $K^{\mathcal{E}}_{i-1}(S_{i-1})=$
$$
 \sum_{S_i} \mathbb{P}_{\theta_1}(S_i | S_{i-1},\E) \left(\log
\frac{\mathbb{P}_{\theta_1}(S_i | S_{i-1},\E)}{\mathbb{P}_{\theta_0}(S_i | S_{i-1},\E)}
+K^{\mathcal{E}}_{i}(S_i)
\right),
$$
for $i=n,\ldots,2$; $D(\PP_{\theta_1}(S|\mathcal{E}) ||\, \PP_{\theta_0}(S|\mathcal{E}))=$
$$
\sum_{S_1}\PP_{\theta_1}(S_1|\E)\left(\log\frac{\PP_{\theta_1}(S_1|\E)}{\PP_{\theta_0}(S_1|\E)} + K^{\mathcal{E}}_1(S_1)\right).
$$

The conditional probabilities $\PP(S_i|S_{i-1},\E)$ are computed recursively \cite{Rabiner89atutorial}: for instance, one can consider the \textit{backward quantities} $B_i(s) = \PP(X_{i+1:N}=x_{i+1:N}|S_i=s)$.
In the homogeneous case\footnote{In the heterogeneous case we can classically derive similar formulas.}, these are computed recursively from $B_n(s)=1$ with
$B_{i-1}(r) = \sum_{s}\pi(r,s)e(s,x_i)B_i(s)$, for $i=n,\ldots,2$.
Then we obtain the following conditional probabilities:
$\PP(S_i =s|S_{i-1} = r, \mathcal{E}) = \pi(r,s)e(s,x_i)B_i(s)/B_{i-1}(r)$, and
$\PP(S_1 =s |\E) \propto \mu(s)e(s,x_1)B_1(s)$.


\section{Numerical Experiments}

We ran numerical experiments to compare our exact formulas with Monte Carlo approximations.


\textbf{HMTs, no evidence.} We compared the exact value and Monte Carlo estimations of the KLD for the pair of trees considered in \cite{do2003fast}. In these trees, the variables of type $X$ are mixtures of two zero-mean Gaussians: we can easily adapt Eq.~(\ref{eq:tree_inward_initial}) to this case as shown in the Supplementary Material. The exact value of the KLD is 0.690. The results in Table \ref{tab} show that an important number of simulations is necessary for the MC estimations to approximate properly the exact KLD value. We computed the bound suggested by Do in \cite{do2003fast} and obtained a value which is different from the one shown in Figure~3 of \cite{do2003fast}; in particular the value of Do's bound turned out to be the same as the value of the exact KLD. This inconsistency is probably due to a minor numerical issue in \cite{do2003fast} and can be safely ignored because Monte Carlo estimations clearly validate our computations.

\begin{table}[t!]
\caption{HMTs with no evidence, exact KLD = 0.690.} 
\label{tab}
$$
\begin{array}{|c|c|c|}
\hline
\mbox{Trials} & \mbox{MC} & 95\% \mbox{ CI} \\
\hline
10^2 & 0.752  & [0.580,0.925] \\
\hline
10 ^3 & 0.673 & [0.616,0.730] \\
\hline
10^4 & 0.691 & [0.673,0.709] \\
\hline
10^5 & 0.690 & [0.684,0.696] \\
\hline
10^6 & 0.688 & [0.687,0.690] \\
\hline
\end{array}
$$
\end{table}


\textbf{HMMs, no evidence.} We experimented with the pair of discrete HMMs considered in \cite{do2003fast}, the two sets of parameters can be found in the Supplementary Material. We implemented Eqs.~(\ref{eq:kld_sum_hmm}),~(\ref{eq:kldr}) for computing $D(\theta_1||\theta_0)/N$ and the KLDR. For Monte Carlo estimations, we ran $n=1000$ independent trails for each value of $N$. The results are depicted in Figure \ref{fig:hmm_no_evidence_do} and show that the proposed recursions for the computation of the exact KLD give consistent results with Monte Carlo approximations. Moreover the ratio $D(\theta_1||\theta_0)/N$ converges very fast to the KLD rate. Note that these results differ from the ones in Figure~2 of \cite{do2003fast} where Do's bound (i.e. the exact KLD rate) seems not to be attained for $N=100$. Again, Monte Carlo estimations support our computations.


\begin{figure}[!t]
\centering
\includegraphics[width=0.44\textwidth]{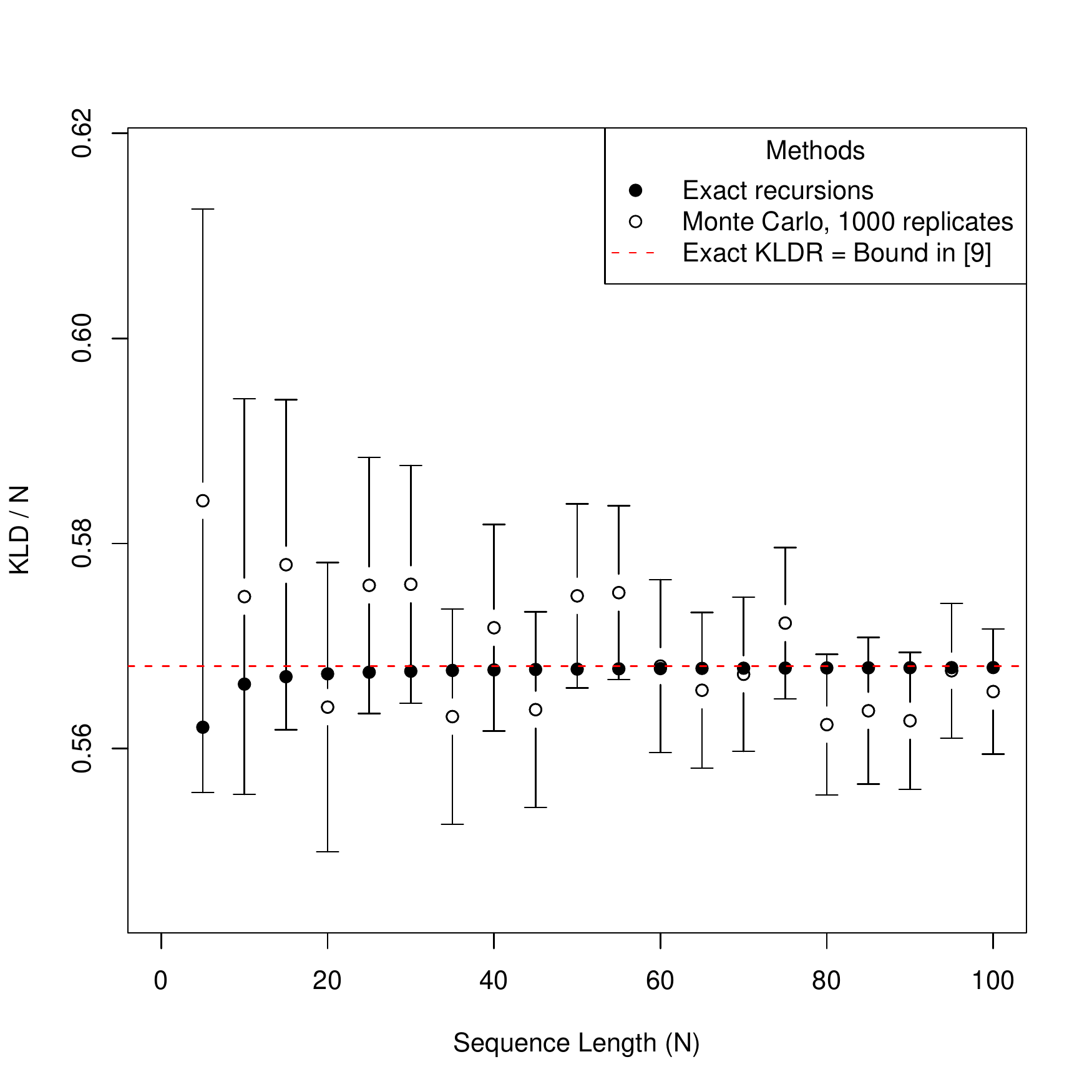}
\caption{HMMs with no evidence. 95\% confidence intervals shown for MC estimations.}
\label{fig:hmm_no_evidence_do}
\end{figure}

\textbf{HMMs with evidence.} We considered the same HMMs as above with an arbitrarily given evidence $\mathcal{E}=\{X_{1:N} = x_{1:N}\}$ (see the Supplementary Material). Figure \ref{fig:hmm_evidence} shows that the exact values of $D(\PP_{\theta_1}(S||\mathcal{E})||\PP(S_{\theta_0}|\mathcal{E}))$ computed with our recursions are consistent with Monte Carlo approximations. In this case, there is no asymptotical behavior because of the irregularity of the evidence. 

\begin{figure}[!t]
\centering
\includegraphics[width=0.44\textwidth]{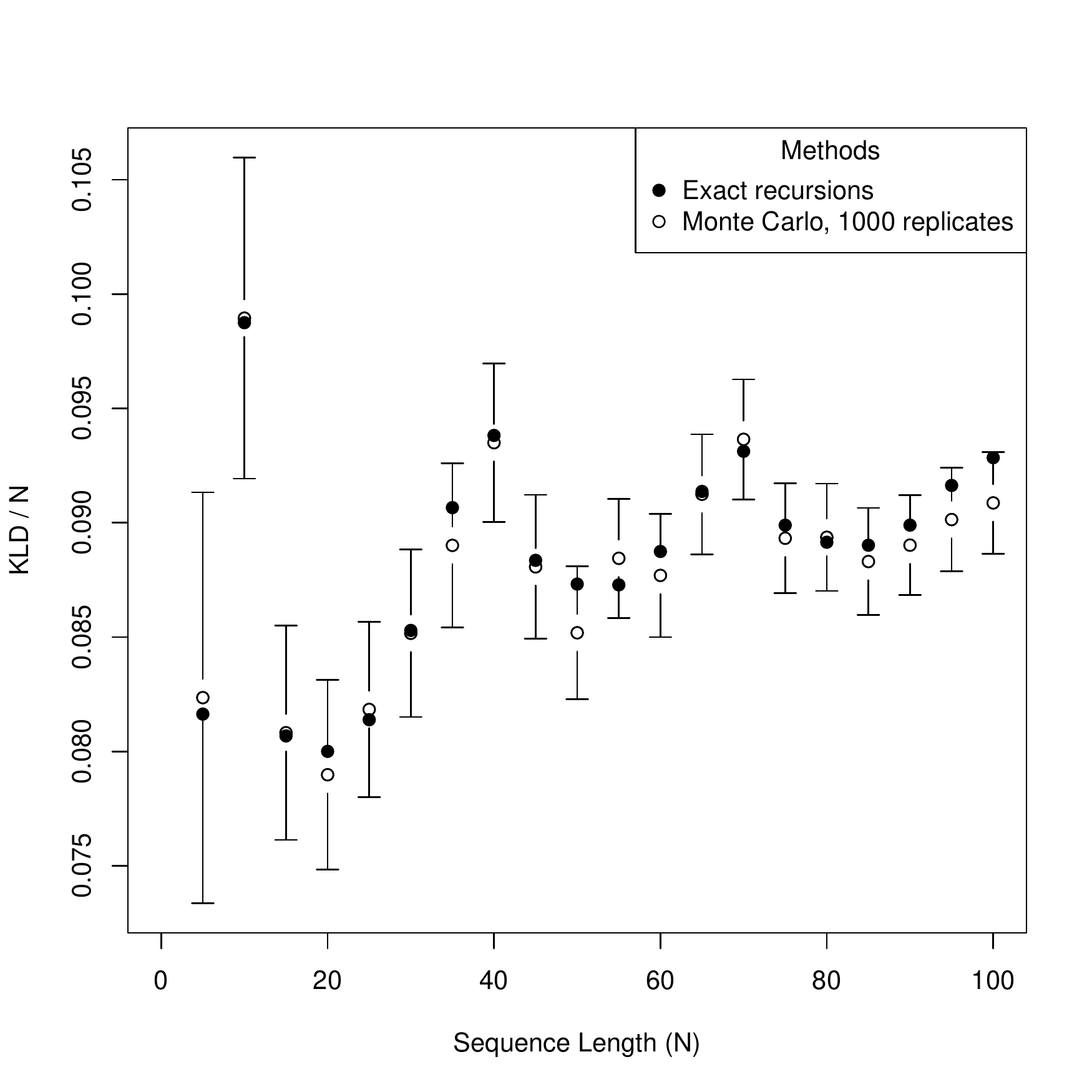}
\caption{HMMs conditioned on the observable variables. 95\% confidence intervals shown for MC estimations.}
\label{fig:hmm_evidence}
\end{figure}

\section{Conclusion}
The most important contribution of this paper is a new theoretical framework for the exact computation of the Kullback-Leibler distance between two hidden Markov trees (or models) based on backward recursions.  
This approach makes it possible to obtain new recursive formulas for computing the exact distance between the conditional probabilities of two hidden Markov models when the observable variables are given as an evidence. When no evidence is given, we derive a closed-form expression for the exact value of the KLD  which generalizes previous results about Markov chains \cite{rached2004kullback}. In the case of HMMs this generalization is not surprising as the pairs of hidden and observable variables are the elements of a Markov chain. However, quite surprisingly, at the best of our knowledge these results have not been explicitly derived earlier. 

It can be easily shown that our closed-form expression is exactly the bound suggested in \cite{do2003fast}: the proof for HMMs is given in the Supplementary Material. In \cite{do2003fast} a necessary and sufficient condition is given for the bound to be the exact value of the KLD. We argue that the suggested bound is the exact value even if this condition is not satisfied (a simple numerical counterexample is given in the Supplementary Material). The reason why the exact value of the KLD is considered as an upper bound in \cite{do2003fast} seems to be an inappropriate use of the equality condition in Lemma 1 \cite{do2003fast}. Indeed this condition is certainly sufficient but not necessary (because $\int f = \int g$ does not imply $f=g$). At last, we observe that the main difference between our formalism and the one in \cite{do2003fast} is that we suggest new recursions to compute the KLD, whereas in \cite{do2003fast} the standard backward quantities for HMTs and HMMs are used.



\section*{Acknowledgement}
V.P. is supported by the \textit{Fondation Sciences Math\'ematiques de Paris}  postdoctoral fellowship program, 2011-2013.


%




\ifCLASSOPTIONcaptionsoff
  \newpage
\fi



%


\bibliographystyle{IEEEtran}
\bibliography{kullback}

\clearpage



\appendix[Supplementary Material with Technical Details]

\section*{Proofs of Results in Section II}
We give detailed proofs of some of the results from Section~II in the main paper. 
\medskip

\begin{IEEEproof}[Proof of Theorem 1] Eq.~(2) in the main paper is obtained from Eq.~(1) by observing that $\PP(X_{ua-},S_{ua-}|S_u) =$
$$
\PP(X_{ua},S_{ua}|S_u)\PP(X_{uab-},S_{uab-} \mbox{ for all }  b\in V|S_{ua}),
$$
and that $\bigcup_{b\in V}\{X_{uab-},S_{uab-}\}$ is a partition of $\{X_{ua-},S_{ua-}\}-\{X_{ua},S_{ua}\}$.  
\end{IEEEproof}

\medskip

In order to prove Corollary 2, first we prove the following lemma: 
\begin{lemma}
If the transition and emission probabilities are the same across the whole tree, then:
\label{prop:tree_inward}
$$
K_{ua\rightarrow u}(S_u) = k_{ua}(S_u) + \sum_{S_{ua}}\PP_{\theta_1}(S_{ua}|S_u)\sum_{b\in V}K_{uab\rightarrow ua}(S_{ua}),
$$
where $k_{ua}(S_u)$ does not depend on $u$ and $a$: $k_{ua}(S_u=r)=$ 
$$
\sum_{x,s} \pi_{\theta_1}(r,s)e_{\theta_1}(s,x)\log\frac{ \pi_{\theta_1}(r,s)e_{\theta_1}(s,x)}{ \pi_{\theta_0}(r,s)e_{\theta_0}(s,x)} =
$$
$D[\PP_{\theta_1}(X_0,S_0|S_{\emptyset}=r)|| \PP_{\theta_0}(X_0,S_0|S_{\emptyset}=r)]=k(r)$.
Moreover
$$
D(\theta_1|| \theta_0) = k_{\emptyset} + \sum_{S_{\emptyset}}\PP_{\theta_1}(S_{\emptyset})\sum_{a\in V}K_{a\rightarrow\emptyset}(S_{\emptyset}),
$$
where 
$$
k_{\emptyset} = \sum_{x,s} \mu_{\theta_1}(s)e_{\theta_1}(s,x)\log\frac{ \mu_{\theta_1}(s)e_{\theta_1}(s,x)}{ \mu_{\theta_0}(s)e_{\theta_0}(s,x)}=
$$
$D[\PP_{\theta_1}(X_{\emptyset},S_{\emptyset})|| \PP_{\theta_0}(X_{\emptyset},S_{\emptyset})]$.
\end{lemma}
\medskip
\begin{IEEEproof}[Proof of Lemma 1] 
We only prove the first equation. Because of Eq.~(2): $K_{ua\rightarrow u}(S_u)=$
\begin{eqnarray*}
\sum_{X_{ua},S_{ua}} & \PP_{\theta_1}(X_{ua},S_{ua}|S_u)\log\frac{\PP_{\theta_1}(X_{ua},S_{ua}|S_u)}{\PP_{\theta_0}(X_{ua},S_{ua}|S_u)} + \\
\sum_{X_{ua},S_{ua}} & \PP_{\theta_1}(X_{ua},S_{ua}|S_u)\sum_{b\in V}K_{uab\rightarrow ua}(S_{ua}) 
\end{eqnarray*}
The first term in this sum does not depend on $u$ and $a$ since the transition and emission probabilities are constant; it is straightforward to obtain its expression $k(\cdot)$. The second term is equal to
\begin{multline*}
\sum_{S_{ua}} \PP_{\theta_1}(S_{ua}|S_u)\sum_{b\in V} K_{uab\rightarrow ua}(S_{ua})\sum_{X_{ua}}\PP_{\theta_1}(X_{ua}|S_{ua}) = \\
\sum_{S_{ua}} \PP_{\theta_1}(S_{ua}|S_u)\sum_{b\in V} K_{uab\rightarrow ua}(S_{ua}).
\end{multline*}
\end{IEEEproof}

\begin{IEEEproof}[Proof of Corollary 2]
The key point here is to prove that  for each $a, a' \in V$: $K_{ua\rightarrow u}(S_u) = K_{ua'\rightarrow u}(S_u)$; we will do it by induction on the levels of the tree. By definition of inward quantity and by the lemma above, if $X_{ua}$ is a leaf, with $a\in V$, then $K_{ua\rightarrow u}(S_u)=k(S_u)$ for each $a$. 

Now suppose that 
$K_{uab\rightarrow ua}(S_{ua}) = K_{uab'\rightarrow ua}(S_{ua})$ $\forall\, a,b,b'\in V$ and $\forall u$ of a given length $m$ (\emph{inductive step}). In particular, for each $a,b\in V$ and $u$ of length $m$, we have $K_{uab\rightarrow ua}(S_{ua}) = K_{u00\rightarrow u0}(S_{u0})$. It is now easy to see that 
$K_{ua\rightarrow u}(S_u)=K_{ua'\rightarrow u}(S_u)$ for each $a,a'\in V$ and $u$ of length $m$: by the lemma above
\begin{multline*}
K_{ua\rightarrow u}(S_u) = k(S_u) + \sum_{S_{ua}}\PP_{\theta_1}(S_{ua}|S_u)\sum_{b\in V}K_{uab\rightarrow ua }(S_{ua})\\ =k(S_u) + \sum_{S_{u0}}\PP_{\theta_1}(S_{u0}|S_u)\sum_{b\in V}K_{u00\rightarrow u0}(S_{u0}) \\
=k(S_u) + C\sum_{S_{u0}}\PP_{\theta_1}(S_{u0}|S_u)K_{u00\rightarrow u0}(S_{u0}).
\end{multline*}
\end{IEEEproof}

\section*{Comparison with~\cite{do2003fast}}
We show that the bound suggested by Do in~\cite{do2003fast} is the actual value of the KLD. For the sake of simplicity we will only consider HMMs, however it is straightforward to generalize the following to more general HMTs. 

The closed form expression for the exact value of the KLD between HMMs (no evidence) is 
\begin{equation*}
\label{eq:kld_sum_hmm}
D(\theta_1||\theta_0) = k_\emptyset + \boldsymbol{\mu}_{\theta_1} ( \boldsymbol{I} + \boldsymbol{\pi}_{\theta_1} + \ldots +\boldsymbol{\pi}^{N-2}_{\theta_1}) \boldsymbol{k}.
\end{equation*}
For comparison purposes, we rewrite $k_\emptyset$ and $\boldsymbol{k}$ as
$$
k_{\emptyset} = D(\boldsymbol{\mu}_{\theta_1}||\boldsymbol{\mu}_{\theta_0}) + \boldsymbol{\mu}_{\theta_1}D(\boldsymbol{e}_{\theta_1}||\boldsymbol{e}_{\theta_0}) = D(\boldsymbol{\mu}) + \boldsymbol{\mu}_{\theta_1}D(\boldsymbol{e})
$$
$$
\boldsymbol{k} = D(\boldsymbol{\pi}_{\theta_1}||\boldsymbol{\pi}_{\theta_0}) +  \boldsymbol{\pi_{\theta_1}}D(\boldsymbol{e}_{\theta_1}||\boldsymbol{e}_{\theta_0}) = D(\boldsymbol{\pi}) +  \boldsymbol{\pi}_{\theta_1}D(\boldsymbol{e}),
$$
where the $j$th component of the vector $D(\boldsymbol{e}):=D(\boldsymbol{e}_{\theta_1}||\boldsymbol{e}_{\theta_0})$ is $D(\boldsymbol{e}_{\theta_1}(j,\cdot)||\boldsymbol{e}_{\theta_0}(j,\cdot))$, and similarly  the $j$th component of $D(\boldsymbol{\pi}):=D(\boldsymbol{\pi}_{\theta_1}||\boldsymbol{\pi}_{\theta_0})$ is $D(\boldsymbol{\pi}_{\theta_1}(j,\cdot)||\boldsymbol{\pi}_{\theta_0}(j,\cdot))$. The reader should not confound Do's symbol $\boldsymbol{e}$, which is $D(\boldsymbol{e})$ in our notations, with our emission matrix $\boldsymbol{e}$. Moreover Do's vector $\boldsymbol{d}$ becomes $D(\boldsymbol{\pi}) + D(\boldsymbol{e})$ in our notations.

Using these notations, Do's upper bound in the case of HMMs - Eq.~(19) in~\cite{do2003fast} - is $U=$
\begin{equation*}
\label{eq:dobound}
D(\boldsymbol{\mu}) + \boldsymbol{\mu}_{\theta_1} \left( \sum_{i=1}^{N-1} \boldsymbol{\pi}_{\theta_1}^{i-1}[D(\boldsymbol{\pi}) + D(\boldsymbol{e})]  + \boldsymbol{\pi}_{\theta_1}^{N-1}D(\boldsymbol{e}) \right)
\end{equation*}
\begin{proposition}
$D(\theta_1||\theta_0)=U$.
\end{proposition}

\begin{proof}
\begin{multline*}
D(\theta_1||\theta_0)=D(\boldsymbol{\mu}) + \boldsymbol{\mu}_{\theta_1}D(\boldsymbol{e})\, +  \\ \boldsymbol{\mu}_{\theta_1}( \boldsymbol{I} + \boldsymbol{\pi}_{\theta_1} + \ldots +\boldsymbol{\pi}^{N-2}_{\theta_1}) 
(D(\boldsymbol{\pi})+\boldsymbol{\pi}_{\theta_1}D(\boldsymbol{e})) = \\
D(\boldsymbol{\mu}) + \boldsymbol{\mu}_{\theta_1} \left( \sum_{i=1}^{N-1} \boldsymbol{\pi}_{\theta_1}^{i-1}[D(\boldsymbol{\pi}) + D(\boldsymbol{e})]  + \boldsymbol{\pi}_{\theta_1}^{N-1}D(\boldsymbol{e}) \right).
\end{multline*}
\end{proof}

In~\cite{do2003fast} it is explained that $D(\theta_1||\theta_0)=U$ if and only if 
$$
\PP_{\theta_1}(S=s|X=x) = \PP_{\theta_0}(S=s|X=x), \mbox{ for all } s,x.
$$

We observe that this condition is not fulfilled in general, whereas $D(\theta_1||\theta_0)=U$ is always true as shown above. For instance, consider the HMMs of length 10 with the same parameters as in Eq.~(22)~\cite{do2003fast}. For the arbitrarily fixed
$$
s = (1,1,1,1,1,2,2,2,2,2)
$$
$$
x = (1,1,1,2,2,2,3,3,3,3)
$$
we have
$\PP_{\theta_1}(S=s|X=x)=0.91,\,\PP_{\theta_0}(S=s|X=x)=0.10$ and $D(\theta_1||\theta_0)=U=0.071$.

\section*{Computation of $\sum_{i=0}^{N-2} \boldsymbol{\pi}^i \boldsymbol{k}$}

When considering HMMs with no evidence, the exact KLD expression involves a term of the form $\sum_{i=0}^{N-2} \boldsymbol{\pi}^i \boldsymbol{k}$ where $\boldsymbol{\pi}$ is a stochastic matrix (of order $d$, where $d$ is the number of hidden states), and $\boldsymbol{k}$ is a column-vector. Note that,  because $\boldsymbol{\pi}$ is stochastic, $\boldsymbol{I}-\boldsymbol{\pi}$ is not invertible. Is it possible to compute this sum with a complexity smaller than $O(d^2 N)$? The answer to this question is indeed ``yes'', but a little bit of linear algebra is required.

Let us assume that there exists $\mathbf{P}=(\mathbf{v}_1,\ldots,\mathbf{v}_d)$ a basis of (column-) eigenvectors of $\boldsymbol{\pi}$ such that $\boldsymbol{\pi}=\mathbf{P}\mathbf{D}\mathbf{P}^{-1}$, where $D=\text{diag}(\lambda_1,\ldots,\lambda_d)$ is the diagonal matrix of the corresponding eigenvalues. For the sake of simplicity, we assume that $\lambda_1=1.0$ and that $|\lambda_j|<1$ if $j \neq 1$ (for example, this is true if $\boldsymbol{\pi}$ is \emph{primitive}, which means that $\exists i$ such as $\boldsymbol{\pi}^i>0$). Nevertheless, the following method can be easily extended to the case when the eigenvalue $1.0$ has a multiplicity greater than $1$.

By defining the invertible matrix $\widetilde{\boldsymbol{\pi}}=\mathbf{P}\text{diag}(0,\lambda_2\ldots,\lambda_d)\mathbf{P}^{-1}$ and decomposing $\boldsymbol{k}$ with respect to the eigenvector basis as $\boldsymbol{k}=k_{1}\mathbf{v}_1 + \widetilde{\boldsymbol{k}}$, we obtain 
$$
\boldsymbol{\pi}\boldsymbol{k}=k_1 \mathbf{v}_1 + \widetilde{\boldsymbol{\pi}} \widetilde{\boldsymbol{k}}.
$$
It follows that
\begin{eqnarray*}
\sum_{i=0}^{N-2} \boldsymbol{\pi}^i \boldsymbol{k}&=&(N-1)k_1 \mathbf{v}_1+\sum_{i=0}^{N-2} \widetilde{\boldsymbol{\pi}}^i \widetilde{\boldsymbol{k}}\\
&=&(N-1)k_1 \mathbf{v}_1 + (\boldsymbol{I}-\widetilde{\boldsymbol{\pi}}^{N-1})(\boldsymbol{I}-\widetilde{\boldsymbol{\pi}})^{-1}\widetilde{\boldsymbol{k}}
\end{eqnarray*}
which can be computed in $O(d^3 \log_2 N)$ by obtaining $\widetilde{\boldsymbol{\pi}}^{N-1}$ through a binary decomposition of $N-1$.

\section*{Numerical Experiments}

\subsection*{HMMs with no evidence}
\label{subset:HMMnoev}
We considered the same set of parameters as in Eq.~(22)~\cite{do2003fast}. In our notations:
$$
\begin{array}{lcl}
\boldsymbol{\mu}_{\theta_1} = 
\begin{array}{cc}
(0.5 & 0.5)
\end{array}
&
\boldsymbol{\mu}_{\theta_0}=  
\begin{array}{cc}
(0.5 & 0.5)
\end{array} \\
\boldsymbol{\pi}_{\theta_1} = 
\left(
\begin{array}{cc}
0.9 & 0.1 \\
0.2 & 0.8 
\end{array}
\right)
&
\boldsymbol{\pi}_{\theta_0} = 
\left(
\begin{array}{cc}
0.7 & 0.3 \\
0.4 & 0.6 
\end{array}
\right)\\
\boldsymbol{e}_{\theta_1}=
\left(
\begin{array}{ccc}
0.1 & 0.3 & 0.6 \\
0.2 & 0.1 & 0.7
\end{array}
\right)
&
\boldsymbol{e}_{\theta_0}=
\left(
\begin{array}{ccc}
0.3 & 0.5 & 0.2 \\
0.6 & 0.2 & 0.2
\end{array}
\right)
\end{array}
$$

The stationary distribution of $\boldsymbol{\pi}_{\theta_1}$ is $\boldsymbol{\nu} = (2/3\,\,\,\,\,1/3)$: $\boldsymbol{\nu}\boldsymbol{\pi}_{\theta_{1}}=\boldsymbol{\nu}$ and $\boldsymbol{\mu}_{\theta_1}\boldsymbol{\pi}^i_{\theta_1}\rightarrow \boldsymbol{\nu}$ for large $i$.

\subsection*{HMMs with evidence}

For $N=100$ we took as evidence the vector $ x_{1:100}$ where: 1) all the components with positions $[1,10], [31,40], [61,70], [91,100]$ are equal to 1; 2) the components with positions $[11,20], [41,50], [71,80]$ are equal to 2; 3) the components with positions $[21,30], [51,60], [81,90]$ are equal to 3. For $5 \leq N\leq 95$, the components of $x_{1:N}$ are the first $N$ values of $x_{1:100}$.

\subsection*{HMTs, no evidence}

We considered the same HMTs as in Eq.~(23)~\cite{do2003fast}. All the $S$ nodes belonging to the same level have the same set of parameters. In our notations:
$$
\begin{array}{lcl}
\boldsymbol{\mu}_{\theta_1} = 
\begin{array}{cc}
(0.69 & 0.31)
\end{array}
&
\boldsymbol{\mu}_{\theta_0}=  
\begin{array}{cc}
(0.63 & 0.37)
\end{array} \\
\boldsymbol{\pi}^{0}_{\theta_1} = 
\left(
\begin{array}{cc}
0.99 & 0.01 \\
0.22 & 0.78 
\end{array}
\right)
&
\boldsymbol{\pi}^0_{\theta_0} = 
\left(
\begin{array}{cc}
0.98 & 0.02 \\
0.20 & 0.80 
\end{array}
\right)\\
\boldsymbol{\pi}^{00}_{\theta_1} = 
\left(
\begin{array}{cc}
0.99 & 0.01 \\
0.32 & 0.68 
\end{array}
\right)
&
\boldsymbol{\pi}^{00}_{\theta_0} = 
\left(
\begin{array}{cc}
0.99 & 0.01 \\
0.22 & 0.78 
\end{array}
\right)
\end{array}
$$
Each emission probability distribution $\PP(X_u|S_u)$ has a zero-mean Gaussian density with standard deviation depending on $S_u$ as follows:
$$
\begin{array}{lcl}
\sigma^{\emptyset}_{\theta_1}(1)=11.8, \sigma^{\emptyset}_{\theta_1}(2)=67.1 && \sigma^{\emptyset}_{\theta_0}(1)=24.6, \sigma^{\emptyset}_{\theta_0}(2)=74.8 
\end{array}
$$
$$
\begin{array}{lcl}
\sigma^{0}_{\theta_1}(1)=4.1, \sigma^{0}_{\theta_1}(2)=29.3 && \sigma^{0}_{\theta_0}(1)=6.9, \sigma^{0}_{\theta_0}(2)=31.9 
\end{array}
$$
$$
\begin{array}{lcl}
\sigma^{00}_{\theta_1}(1)=2.8, \sigma^{00}_{\theta_1}(2)=10.3 && \sigma^{00}_{\theta_0}(1)=3.1, \sigma^{00}_{\theta_0}(2)=14.8
\end{array}
$$
For instance, the probability density function $f_{\theta_1}(X_{10}|S_{10})$ is $\mathcal{N}(0,\sigma_{\theta_1}^{00}(1))$ if $S_{10}=1$ and $\mathcal{N}(0,\sigma_{\theta_1}^{00}(2))$ if $S_{10}=2$.

Eq.~(2) becomes $K_{ua\rightarrow u}(S_u)=$
\begin{multline*}
\sum_{S_{ua}}\PP_{\theta_1}(S_{ua}|S_{u})\left(\log\frac{\PP_{\theta_1}(S_{ua}|S_u)}{\PP_{\theta_0}(S_{ua}|S_u)} \right. + \\ \left. \int_{X_{ua}}f_{\theta_1}(X_{ua}|S_{ua}) \log\frac{f_{\theta_1}(X_{ua}|S_{ua})}{f_{\theta_0}(X_{ua}|S_{ua})}  + \sum_{b\in V}  K_{uab\rightarrow ua}(S_{ua})\right) 
\end{multline*}
\begin{multline*}
= \sum_{S_{ua}}\PP_{\theta_1}(S_{ua}|S_u)\left[\log\frac{\PP_{\theta_1}(S_{ua}|S_u)}{\PP_{\theta_0}(S_{ua}|S_u)}+ \right. \\ D\left[\mathcal{N}(0,\sigma_{\theta_1}^{ua}(S_{ua}))||\mathcal{N}(0,\sigma_{\theta_0}^{ua}(S_{ua}))\right] + \left. \sum_{b\in V}  K_{uab\rightarrow ua}(S_{ua})\right],
\end{multline*}
where $a\in V$. If $X_{ua}$ is a leaf then $K_{uab\rightarrow ua}(S_{ua})=0$.  If $X_{ua}$ is a not leaf then $K_{uab\rightarrow ua}$ does not depend on $b\in V$ and therefore
$$
\sum_{b\in V}  K_{uab\rightarrow ua}(S_{ua}) = 2\cdot K_{ua0\rightarrow ua}(S_{ua}).
$$
Similarly, one can obtain the formula for the KLD. 

At last, we recall that the KLD between two Gaussians can be computed with the well known formula 
\begin{multline*}
D(\mathcal{N}(\mu_1,\sigma_1)||\mathcal{N}(\mu_1,\sigma_1))= \\
\frac{\sigma_1^2+(\mu_1-\mu_0)^2}{2\sigma_0^2}+ \log\frac{\sigma_0}{\sigma_1} - \frac{1}{2}.
\end{multline*}

\end{document}